\documentclass{amsart}
\usepackage[all]{xy}   
\usepackage {amssymb,latexsym,amsthm,amsmath}
\usepackage{hyperref}
\usepackage{graphicx}
\usepackage{xcolor}
\usepackage[utf8]{inputenc}
\usepackage{palatino}
\usepackage{adjustbox}
\usepackage{tikz}
\usepackage{multirow}
\usepackage{float}

\usepackage[pagewise]{lineno}\nolinenumbers
\makeatletter
\@namedef{subjclassname@2020}{%
  \textup{2020} Mathematics Subject Classification}
\makeatother

\long\def\symbolfootnote[#1]#2{\begingroup%
\def\thefootnote{\fnsymbol{footnote}}\footnote[#1]{#2}\endgroup}

\newcommand{\N}{\mathbb N}

\newcommand{\C}{\mathcal C}

\newcommand{\F}{\mathbb F}

\makeatletter
\def\imod#1{\allowbreak\mkern10mu({\operator@font mod}\,\,#1)}
\makeatother
\theoremstyle{plain}
\newtheorem{thm}{Theorem}[section]
\theoremstyle{definition}
\newtheorem{defn}{Definition}[section]
\newtheorem{lemma}{Lemma}[section]
\newtheorem{prop}{Proposition}[section]
\newtheorem{cor}{Corollary}[section]
\newtheorem{remark}{Remark}[section]
\newtheorem{example}{Example}[section]
\numberwithin{equation}{section}
\newcommand{\ignore}[1]{}

\newcommand{\mynote}[1]{}

\begin{document}
\title[A Class of $(n, k, r, t)_i$ LRCs Via Parity Check Matrix]{A Class of $(n, k, r, t)_i$ LRCs Via Parity Check Matrix}

\author[Mukhopadhyay]{Deep Mukhopadhyay}
\author[Bhowmick]{Sanjit Bhowmick}
\author[Hansda]{Kalyan Hansda}
\author[Bagchi]{Satya Bagchi}

\address[Mukhopadhyay]{Department of Mathematics, National Institute of Technology Durgapur, Durgapur-713209, India.}
\email{deepmukh1998@gmail.com}

\address[Bhowmick]{Department of Mathematics, Indraprastha Institute of Information Technology Delhi, New Delhi-713209, India.}
\email{sanjitbhowmick392@gmail.com}

\address[Hansda]{Department of Mathematics, Visva-Bharati University, Santiniketan, Bolpur-731235, West Bengal, India.}
\email{kalyanh4@gmail.com}

\address[Bagchi]{Department of Mathematics, National Institute of Technology Durgapur, Durgapur-713209, India.}
\email{satya.bagchi@maths.nitdgp.ac.in}


\maketitle

\begin{abstract}
A code is called $(n, k, r, t)$ information symbol locally repairable code \big($(n, k, r, t)_i$ LRC\big) if each information coordinate can be achieved by at least $t$ disjoint repair sets, containing at most $r$ other coordinates. This paper considers a class of $(n, k, r, t)_i$ LRCs, where each repair set contains exactly one parity coordinate. We explore the systematic code in terms of the standard parity check matrix. First, some structural features of the parity check matrix are proposed by showing some connections with the membership matrix and the minimum distance optimality of the code. Next to that, parity check matrix based proofs of various bounds associated with the code are placed. In addition to this, we provide several constructions of optimal $(n, k, r, t)_i$ LRCs, with the help of two Cayley tables of a finite field. Finally, we generalize a result of $q$-ary $(n, k, r)$ LRCs to $q$-ary $(n, k, r, t)$ LRCs.
\end{abstract}
\section{Introduction}
Codes for distributed storage systems become a recent advances in Coding Theory. To increase the data reliability, while storing data in distributed storage systems, the method that usually employed is replication. However, in this case, the storage cost is much higher than the data reliability. Thus, instead of replication, coding theory (Erasure Coding) is then used, where the data is divided into $k$ parts and stored into $n$ different data nodes. It is efficient to overcome in terms of both node failures and storage overhead problems, although it suffers from substantial repair costs. To deal this, some coding theory techniques are designed. Locally repairable code is one of them.
 
A code symbol has locality $r$, if it can be fetched by at most $r$ other code symbols. An $[n,k,d]$ linear code with $r$ locality is said to be $(n,k,r)$ locally repairable code (LRC). LRCs with $r$ locality for information symbols are called \emph{information symbol} locally repairable codes \big($(n, k, r)_i$ LRC\big), and if all symbols have locality $r$, then the formatted codes are \emph{all symbol} locally repairable codes \big($(n, k, r)_a$ LRC\big).

Gopalan \textit{et al.} \cite{Gopalan} firstly introduced $(n,k,r)_i$ LRCs and gave lucidly a Singleton like upper bound on the minimum distance which is given by
\begin{equation}\label{LRC_eq1}
d \leq n-k - \left\lceil\frac{k}{r}\right\rceil + 2.
\end{equation} 
Later on, it is seen that the bound (\ref{LRC_eq1}) is true for $(n,k,r)_a$ LRCs also. Following the same approach of Gopalan \textit{et al.}, Forbes and Yekhanin proposed locality and proved the bound (\ref{LRC_eq1}) for non-linear codes in \cite{Forbes}. A family of optimal $(n,k,r)$ LRCs against the bound (\ref{LRC_eq1}) was constructed by Tamo and Barg \cite{Tamo2014}. Hao \textit{et al.} \cite{Hao} overviewed $(n,k,r)$ LRCs via parity check matrix and gave different possible classes of optimal binary $(n,k,r)$ LRCs. In an LRC, the set of other symbols (positions) that help to reconstruct a particular symbol is called the repair set or recovery set of that symbol. For an $(n,k,r)$ LRC, the number of repair sets for a code symbol is one. Subsequently, this notion is generalized by choosing $t$ repair sets for a code symbol, which is the availability of a LRC. LRCs with $t$ availability are referred as $(n, k, r, t)$ LRCs. The notion of availability was initiated by Rawat \textit{et al.} in \cite{Rawat}.

In an $(n, k, r, t)$ LRC, the distinct subsets of the form $\Gamma_j(i) \cup \{i\}$ \big(where $\Gamma_j(i)$ is the $j$th repair set of the $i$th coordinate\big) are called the local groups (subsets) of the code \cite{Rawat}. The $\{0, 1\}$ membership matrix $\textbf{R}$ of order $m \times k$ of an LRC is constructed from the local groups of that code \cite{Rawat}. Rawat \textit{et al.} in \cite{Rawat} also originated a restricted class of $(n, k, r, t)$ LRCs, where there is exactly one parity symbol in each repair set. In the same article, the author placed an upper bound on the minimum distance of that class through the following inequality, that applies to both \emph{all symbol} and \emph{information Symbol} LRC.
\begin{equation}\label{LRC_eq2}
d \leq n-k - \left\lceil \frac{kt}{r}\right\rceil + t + 1.
\end{equation}
Design of the membership matrix is a key tool that helps to construct the above class of $(n, k, r, t)$ LRCs. There are several techniques for constructing the membership matrix for some locality and availability parameters in \cite{Su}. Due to the achievability of the bound (\ref{LRC_eq2}), this special class of $(n, k, r, t)$ LRCs is mostly studied. Some theoretical achievements can be found in \cite{Hao3}. Several constructions on optimal $(n, k, r, t)$ LRCs against the bound (\ref{LRC_eq2}) can be found in \cite{Hao4}, \cite{Hao3}, \cite{Kim}, \cite{Rawat}, \cite{Tan}, and \cite{J. wang} for different classes of parameters.

In support of $(n,k,r,t)_i$ LRCs without the restriction of the parity bits on the repair set, Wang and Zhang \cite{Wang} derived

\begin{equation}\label{LRC_eq3}
d \leq n-k - \left\lceil\frac{t(k-1) +1}{t(r-1) +1}\right\rceil + 2
\end{equation}

and established that, if $n\geq k(rt+1)$ then there are certain codes which satisfy (\ref{LRC_eq3}) with equality. 

In \cite{Tamo2016}, for $(n,k,r,t)_a$ LRCs, an upper bound on the code rate and minimum distance are developed, which are

\begin{equation}\label{LRC_eq4}
\frac{k}{n} \leq \frac{1}{\prod_{j=1}^t(1+\frac{1}{jr})},
\end{equation}

\begin{equation}\label{LRC_eq5}
d \leq n- \sum_{i=0}^{t} \left\lfloor \frac{k-1}{r^i} \right\rfloor.
\end{equation}

Codes satisfying bounds (\ref{LRC_eq3}) and (\ref{LRC_eq5}) with equality can be found in \cite{Wang}, \cite{Wang2}. Several new bounds are developed for codes with strict availability in \cite{Balaji}, which perform tighter than (\ref{LRC_eq4}). Prakash \textit{et al.} \cite{Prakash} presented an upper bound on the code rate of $(n, k, r, t)$ LRCs with sequential recovery for two erasures, which is

\begin{equation}\label{LRC_eq6}
\frac{k}{n} \leq \frac{r}{r+2}.
\end{equation}
Codes that attain the bound (\ref{LRC_eq6}) were constructed in \cite{Kadhe}. Recently, in some theoretical aspects, dual containing LRCs have been studied in \cite{Mukhopadhyay}. Xu \textit{et al.}, in \cite{Xu}, observed complete structures and constructions of optimal quaternary $(r, \delta)$ LRCs via parity check matrix. Parity check matrix is a valuable tool for studying linear codes. Thus, some analogous study via parity check matrix is necessary to understand various intrinsic algebraic and combinatorial properties of $(n, k, r, t)$ LRCs. 

This article focuses only on systematic $(n,k,r,t)_i$ LRCs, where each recovery set has precisely one parity \big(Single-Parity LRCs\big). Amid these codes, without loss of generality, we consider those systematic codes, where the first $k$ coordinates of a codeword indicate the data symbols. The rest of the paper is arranged in the following manner. After giving some background and the pertinent setting in Section II, we analyze some bounds and properties of the code through the standard parity check matrix in Section III. Section IV consists of several constructions of distance optimal and rate optimal codes. Further, Section V shows a property for a distance optimal $q$-ary $(n, k, r, t)_i$ LRCs. The paper concludes in Section VI.
\section{Notations and Preliminaries}

This section recalls some basic notations and preliminaries, which we use in our further discussions.

\begin{itemize}
    \item $\F_{p^m}$ denotes a finite field of size $p^m$, where $p$ is a prime and $m\in \N$.
    \item For a positive integer $m$, $[m]=\{1, 2,\dots, m\}$.
    \item The support set of a codeword $v=(v_1, v_2,\dots, v_m)$ is $Supp(v)=\{i\in [m] : v_i\neq 0\}$.
    \item $\boldsymbol{v}_g$ denotes $g$-tuple row vector whose all entries are $\boldsymbol{v}$.
    \item $\boldsymbol{v}^T_g$ is the transpose of $\boldsymbol{v}_g$.
    \item $X \otimes Y$ means the Kronecker product of two matrices $X$ and $Y$.
\end{itemize}
We now convey the formal definition of an $(n, k, r, t)_i$ LRCs. 
\begin{defn}\cite{Hao2, Rawat}\label{Definition:1}
A linear code is said to be $(n, k, r, t)_i$ LRC if it holds the following properties.
\begin{enumerate}
\item  For the $i$th code symbol, where $i \in [k]$, there exist $t$ codewords $c_1, \dots, c_t$ of $\C^{\perp}$ such that each $c_j$ covers the $i$th symbol.
\item $|Supp(c_j)| \leq r+1, \forall j \in [t]$.
\item $Supp(c_j) \cap Supp(c_l)  = \{i\}, \forall  j \neq l $ and $j,l \in [t] $.
\end{enumerate}
\end{defn}

The above definition conforms to any $(n, k, r, t)_i$ LRCs. However, the following property is also required for a single parity $(n, k, r, t)_i$ LRC.
\begin{enumerate}
    \item [(4)] $\lvert Supp(c_j) \setminus [k] \rvert = 1, \forall j \in [t]$. 
\end{enumerate}

The subsequent result sets up a characterization of an optimal single parity $(n, k, r, t)$ LRCs.
\begin{thm}\label{Theorem 0}\cite{Hao3}
For a single parity $(n, k, r, t)$ LRC $\C$, if \[d=n-k-\left\lceil\frac{kt}{r}\right\rceil+t+1,\]
then the number of columns in the membership matrix of $\C$ is $\left\lceil\frac{kt}{r}\right\rceil$. Moreover, if $r\mid kt$, then the membership matrix is a regular $k \times \frac{kt}{r}$ matrix with uniform row weight $t$ and column weight $r$.
\end{thm}
We now present a basic observation regarding the Cayley tables of a finite field, which will be helpful in our due course of events. Also, it is mentioned that we consider the elements in the same order in the operative rows and operative columns in both the Cayley tables.

\begin{lemma}\label{Lemma 1}
Let us consider the matrices $T^+$ and $T^\star$, obtained from the Cayley tables of $(\F_{p^m}, +)$  and $(\F_{p^m}, \star)$ \big(including 0\big) respectively by removing both the operative rows and columns. Now, if we choose to form Set I (or Set II) with $4$ elements from $T^+$ (or $T^\star$) in such a way that they are situated in the corners of a rectangle. Then Set I and Set II can not be the same. 
\end{lemma}

\begin{proof}
Let us consider that both Set I and Set II are same. Thus, there are elements $a_1, b_1, c_1, d_1$ and $a_2, b_2, c_2, d_2$ of $\F_{p^m}$ with $a_i \neq b_i$ and $c_i \neq d_i;$ for $i=\{1, 2\}$, such that the following two sets are equal.
\[\{a_1+c_1, b_1+c_1, a_1+d_1, b_1+d_1\}~~\text{and}~~ \{a_2\star c_2, b_2\star c_2, a_2\star d_2, b_2\star d_2\}\] Now, for any choice of equality of elements between these two sets, it leads to a contradiction.
\end{proof}

\section{Parity check matrix viewpoint of $(n, k, r, t)_i$ LRCs}

This section is divided into two subsections. First, we set up some characteristics on the standard parity check matrix for the code. Next, by using those properties of the parity check matrix, some new proofs for different upper and lower bounds are given.

\subsection{Structural features of parity check matrix}
\begin{prop}\label{Proposition 1}
Let $\C$ be a systematic $[n,k,d]$ linear code over a finite field $\F_{p^m}$ with parity check matrix $H := [P|I_{n-k}]$, then $\C$ forms an $(n, k, r, t)_i$ LRC if and only if there exists a submatrix $P_1$ of order $l \times k$ of $P$ that satisfies the following properties.
\begin{enumerate}
    \item For each $i\in [k]$, there exist at least $t$ rows $\Omega_1, \Omega_2, \hdots, \Omega_t$ of $P_1$ such that \[Supp(\Omega_i) \cap Supp(\Omega_j) = \{i\};~ \forall i \neq j \in [t].\]
    \item The weight of each row of $P_1$ is at most $r$.
\end{enumerate}
\end{prop} 

\begin{proof}
Let $\C$ be a systematic $(n, k, r, t)_i$ LRC with the parity check matrix $H$. Suppose the encoding function of the code $\C$ is
\[(m_1, \dots, m_k) \mapsto (m_1, \dots, m_k, p_1, \dots, p_{n-k}),\]
where $m_1, \dots, m_k$ are information bits and $p_1, \dots, p_{n-k}$ are parities, each of which is a linear combination of some $m_j$'s. Let $\textbf{x} =(m_1, \dots, m_k, p_1, \dots, p_{n-k})$ be a codeword in $\C$. Now, for the $i$th information symbol $m_i$ of $\textbf{x}$, there are at least $t$ disjoint repair sets $\Gamma_1(i), \dots, \Gamma_t(i)$ of cardinality at most $r$ and each of them contains exactly one parity. Thus, to ensure $t$ availability of $m_i$, $m_i$ involves in at least $t$ parities. Among these $t$ parities, each $p_i$ is a linear combination of at most $r$ information bits including $m_i$ as $m_i$ has locality  $r$. Moreover, as all the $t$ repair sets of $m_i$ are disjoint, these $t$ parities are formed by distinct sets of information bits excluding $m_i$. Similarly, by considering all information symbols, let $l$ be the number of such distinct parities. Consequently, there are $l$ rows of $P$ in $H$ those who forms $P_1$ and due to the above combination of information bits in the parities, each row of $P_1$ has weight at most $r$ and for each information coordinate $m_i$, there are at least $t$ rows of $P_1$ such that for any two distinct rows, their support sets intersect in the $i$th position.\\
Conversely, a parity check matrix $H$ is given for a linear code with a submatrix $P_1$ that satisfies the given two properties. Now, the rows of $P_1$, including the $I_{n-k}$ part in $H$ produces $l$ codewords $c_1, c_2, \hdots, c_l$ of $\C ^\perp$, and hence it can be seen that the two properties of $P_1$ help the code $\C$ to match those three properties of $(n, k, r, t)_i$ LRC, which is given in Definition \ref{Definition:1}. Furthermore, these $l$ codewords of $\C^\perp$ also satisfy the single parity property $\lvert Supp(c_j) \setminus [k] \rvert = 1, \forall j \in [l].$
\end{proof}

The above $P_1$ matrix plays a crucial role in the characterization of the parity check matrix for an $(n, k, r, t)_i$ LRC. In this regard, we start with the following result which shows a connection between the local groups (subsets) and the rows of the $P_1$ matrix. 

\begin{prop}\label{Proposition 2}
For each of the $m$ distinct local groups (subsets) $S_1, S_2, \dots, S_m$ of the form $\Gamma_j(i) \cup \{i\}$, there is a row in $P_1$ whose support set is equal to $S_i \cap [k]$ for some $i \in [m]$ and vice-versa.
\end{prop}

\begin{proof}
$\Gamma_j(i)$ contains the positions of at most $r-1$ information symbols (excluding the single parity bit) that help to construct the $i$th symbol. Thus, the corresponding parity must be a combination of those $r-1$ information symbols as well as the $i$th information symbol. Hence we can say that the rows of $P_1$ are non zero in those $r$ positions indexed by the set $\left\{\Gamma_j(i) \cup \{i\}\right\} \cap [k]$. Similarly, when a codeword is multiplied by $H$, each row of $P_1$ produces a locality relation which helps to recover a code symbol and hence there is a local group consisting the support of the rows of $P_1$.
\end{proof}

\begin{cor}\label{Corollary 1}
By choosing a proper order of the local groups (subsets), for a binary $(n, k, r, t)_i$ LRC, the transpose of $P_1$ is equal to the membership matrix $\textbf{R}$.
\end{cor}

\begin{cor}\label{Corollary 2}
The number of rows in the $P_1$ matrix equals the number of local groups (subsets) of $\C$, that is, $l=m$. Again from \cite[Lemma 1]{Rawat}, we know that $m \geq \left\lceil\frac{kt}{r}\right\rceil$. Hence, $l \geq \left\lceil\frac{kt}{r}\right\rceil$.
\end{cor}

The following result attributes another necessary and sufficient characterization of the parity check matrix of an $(n, k, r, t)_i$ LRC under some parametric aspect. However, the sufficient part can be found in \cite{Hao3} in a different perspective.

\begin{thm}\label{Theorem 1}
Let $\C$ be a systematic $(n, k, r, t)_i$ LRC with $l <2t$ or $kt + 2 > lr$. Then, $H\equiv\left[\dfrac{P_1}{P_2}\vline I_{n-k} \right]$ becomes a parity check matrix of $\C$ if and only if the following properties holds.
\begin{enumerate}
    \item The weight of each row of $P_1$ is at most $r$.
    \item Each column of $P_1$ has weight at least $t$.
    \item $\lvert Supp(\Omega_i) \cap Supp(\Omega_j)\rvert \leq 1$, where $\Omega_i$ and $\Omega_j$ are any two distinct rows of $P_1$.
\end{enumerate}
\end{thm}

\begin{proof}
Let $\C$ be a systematic $(n, k, r, t)_i$ LRC with $l <2t$ or $kt + 2 > lr$ and let the parity check matrix of $\C$ is $H$. Now, using the similar argument of Proposition \ref{Proposition 1}, the first two properties can be easily proved. For the later part, suppose $\Omega_i$ and $\Omega_j$ are the $i$th and $j$th rows of $P_1$ such that $Supp(\Omega_i) \cap Supp(\Omega_j)$ has at least two points. Choose any two of them, say $x$ and $y$. Thus, the $P_1$ matrix contains $4$ non zero entries in $(i, x),$ $(i, y),$ $(j, x)$ and $(j, y)$ positions. Now, as $x, y \in [k]$, there are at least $t$ rows of $P_1$ such that for any two distinct rows, their support sets only intersect in the $x$th position. Likewise, for $y$, there are also at least $t$ rows of $P_1$ such that for any two distinct rows, their support sets only intersect in the $y$th position. To satisfy this, $x$th and $y$th columns of $P_1$ must have weight at least $t+1$. Moreover, there must be at least $t-1$ non zero positions of the $x$th and the $y$th columns that are different. Accordingly, there should be at least $2t$ rows in $P_1$, which is a contradiction since $l < 2t$. Again, if the $x$th and $y$th columns of $P_1$ have weight at least $t+1$, then it must satisfies $kt + 2 \leq lr$, which again leads to a contradiction. Hence, the result holds.\\
The converse part can be proved by using the similar argument of Proposition \ref{Proposition 1}.
\end{proof}

The parametric conditions cannot be dropped to prove the necessary part of the above result. We furnish an example to support our statement.

\begin{example}
Consider the following matrix $A$ 
\[
A=
\begin{bmatrix}
1 & 1 & 1 & 0 & 1 & 0 & 0 & 0 & 0 & 0 & 0\\
1 & 1 & 0 & 1 & 0 & 1 & 0 & 0 & 0 & 0 & 0\\
1 & 0 & 0 & 0 & 0 & 0 & 1 & 0 & 1 & 1 & 0\\
0 & 1 & 0 & 0 & 0 & 0 & 1 & 1 & 0 & 0 & 1\\
0 & 0 & 1 & 1 & 0 & 0 & 0 & 1 & 1 & 0 & 0\\
0 & 0 & 0 & 0 & 1 & 1 & 0 & 0 & 0 & 1 & 1\\
\end{bmatrix}.
\]
Now, $[A|I_6]$ forms a parity check matrix of a systematic $(17, 11, 4, 2)_i$ binary LRC with the following configuration of the repair sets for each information symbol.

\begin{linenomath}
\begin{align*}
\Gamma_1(1) & =\{2, 3, 5, 12\}, & \Gamma_2(1) & =\{7, 9, 10, 14\};\\
\Gamma_1(2) & =\{1, 4, 6, 13\}, & \Gamma_2(2) & =\{7, 8, 11, 15\};\\
\Gamma_1(3) & =\{1, 2, 5, 12\}, & \Gamma_2(3) & =\{4, 8, 9, 16\};\\
\Gamma_1(4) & =\{1, 2, 6, 13\}, & \Gamma_2(4) & =\{3, 8, 9, 16\};\\
\Gamma_1(5) & =\{1, 2, 3, 12\}, & \Gamma_2(5) & =\{6, 10, 11, 17\};\\
\Gamma_1(6) & =\{1, 2, 4, 13\}, & \Gamma_2(6) & =\{5, 10, 11, 17\};\\
\Gamma_1(7) & =\{1, 9, 10, 14\}, & \Gamma_2(7) & =\{2, 8, 11, 15\};\\
\Gamma_1(8) & =\{2, 7, 11, 15\}, & \Gamma_2(8) & =\{3, 4, 9, 16\};\\
\Gamma_1(9) & =\{1, 7, 10, 14\}, & \Gamma_2(9) & =\{3, 4, 8, 16\};\\
\Gamma_1(10) & =\{1, 7, 9, 14\}, & \Gamma_2(10) & =\{5, 6, 11, 17\};\\
\Gamma_1(11) & =\{2, 7, 8, 15\}, & \Gamma_2(11) & =\{5, 6, 10, 17\}.
\end{align*}
\end{linenomath}
Here $l=6= \left\lceil\frac{kt}{r}\right\rceil$ and it doesn't satisfy either $l <2t$ or $kt + 2 > lr$. Now, it can be easily observed that the support sets of the $1$st and $2$nd row of $A$ intersect in $2$ positions.
\end{example}

The above example also justifies that mere equality of $l$ with $\left\lceil\frac{kt}{r}\right\rceil$ cannot achieve the result. To get the necessary part, we need an additional divisibility condition with $l = \left\lceil\frac{kt}{r}\right\rceil$. We now discuss it more in detail and record it as the corollary of Theorem \ref{Theorem 1}.
\begin{cor}\label{Corollary 3}
Let $\C$ be a systematic $(n , k, r, t)_i$ LRC with $l= \frac{kt}{r}$ and $r\mid kt$. Now, $H\equiv\left[\dfrac{P_1}{P_2}\vline I_{n-k} \right]$ becomes a parity check matrix of $\C$ if and only if the following properties holds.
\begin{enumerate}
    \item Each row weight of $P_1$ is exactly $r$.
    \item The weight of each column of $P_1$ is exactly $t$.
    \item $\lvert Supp(\Omega_i) \cap Supp(\Omega_j)\rvert \leq 1$, where $\Omega_i$ and $\Omega_j$ are any two distinct rows of $P_1$.
\end{enumerate}
\end{cor}

\begin{proof}
The proof directly follows from Theorem \ref{Theorem 0} and Theorem \ref{Theorem 1}
\end{proof}

For an $(n, k, r, t)_i$ LRC $\C$, the condition $l=\frac{kt}{r}$ with $r \mid kt$ also provides the following result, which concerns the size of the information bits in $\C$.

\begin{prop}
For an $(n, k, r, t)_i$ LRC, if $l=\frac{kt}{r}$ with $r\mid kt$, then \[k \geq t(r-1) +1.\]
\end{prop}
\begin{proof}
From Corollary \ref{Corollary 3}, we get that if $l=\frac{kt}{r}$ with $r\mid kt$, then each row and column weight of $P_1$ is exactly $r$ and $t$ respectively. Also, this implies that each information symbol has exactly $t$ repair sets containing exactly $r$ other symbols. Thus, there must be at least $t(r-1)$ distinct information symbols to construct a repair set for an information symbol since each repair set has exactly one parity. Therefore, $k\geq t(r-1)+1$ and the proof follows.
\end{proof}

Using the above results, we now propose the structural form of the standard parity check matrix of an optimal systematic $(n, k, r, t)_i$ LRC.

\begin{thm}\label{Theorem 2}
For an optimal $(n , k, r, t)_i$ LRC $\C$ with $r\mid kt$, the parity check matrix $H\equiv\left[\dfrac{P_1}{P_2}\vline I_{n-k} \right]$ of $\C$ satisfies the following three properties.
\begin{enumerate}
    \item The weight of each row of $P_1$ is exactly $r$.
    \item Each column weight of $P_1$ is exactly $t$.
    \item $\lvert Supp(\Omega_i) \cap Supp(\Omega_j)\rvert \leq 1$, where $\Omega_i$ and $\Omega_j$ are any two distinct rows of $P_1$.
\end{enumerate}
\end{thm}

\begin{proof}
Suppose $\C$ is an optimal $(n, k, r, t)_i$ LRC with $r\mid kt$. Now, from Theorem \ref{Theorem 0}, we can say $l=\frac{kt}{r}$ and hence from Corollary \ref{Corollary 3}, the result holds.
\end{proof}

\subsection{Review of several bounds using parity check matrix}

As said earlier, we now propose some purely parity check matrix based proofs for several bounds of the code. We begin our study with the following well known bound.

\begin{thm}
For an $(n, k, r, t)_i$ LRC $\C$, the minimum distance (Hamming) $d$ satisfies \[d \leq n-k - \left\lceil\frac{kt}{r}\right\rceil + t +1.\]
\end{thm}

\begin{proof}
Let $\C$ be an $(n, k, r, t)_i$ LRC with parity check matrix $H\equiv\left[\dfrac{P_1}{P_2}\vline I_{n-k} \right]$. Now, as the availability of the code is $t$, thus there is an information bit which has exactly $t$ repair sets and hence there is a column of $P_1$ whose weight is exactly $t$. Let it be the $j$th column. Therefore, in $P_1$, the $j$th column has exactly $l-t$ zeros. Accordingly the $j$th column has at most $n-k-l+t$ nonzeros in  $\dfrac{P_1}{P_2}$. Thus, due to $I_{n-k}$, to achieve a linearly dependent set of columns, we have to add at most $n-k-l+t$ columns with the $j$th column. Hence, the minimum distance of $\C$ cannot exceed $n-k-l+t +1$. Consequently, \[d \leq n-k-l+t+1 \leq n-k- \left \lceil \frac{kt}{r}\right\rceil + t + 1.\] The second inequality holds due to Corollary \ref{Corollary 2}. This completes the proof.
\end{proof}
The next inequality describes a lower bound on the length of an $(n, k, r, t)_i$ LRC.

\begin{thm}
For any $(n, k, r, t)_i$ LRC,
\[n \geq k+\left\lceil\frac{kt}{r}\right\rceil.\]
\end{thm}

\begin{proof}
By using the fact that the rows of $P_1$ must lie in $H$ and since $l \geq \left\lceil\frac{kt}{r}\right\rceil$, so accordingly it holds.
\end{proof}
The above result immediately facilitates the following corollary associated with the code rate of an $(n, k, r, t)_i$ LRC.
\begin{cor}
The code rate of an $(n, k, r, t)_i$ LRC is upper bounded by
\[ \frac{k}{n} \leq \left\{
\begin{array}{ll}
\dfrac{k}{k+\left\lceil\frac{kt}{r}\right\rceil}, & \text{ if~~} r\nmid kt;\\
\dfrac{r}{r+t}, & \text{ if~~} r \mid kt.
\end{array}
\right.\]
\end{cor}

\begin{remark}
Note that the code rate is comparatively higher when $r\mid kt$. Furthermore, an $(k+\frac{kt}{r}, k, r, t)_i$ LRC, where $r\mid kt$, is always distance optimal and has the highest code rate among all possible values of the parameters of $(n, k, r, t)_i$ LRCs.
\end{remark}


\section{Construction of binary optimal $(n, k, r, t)_i$ LRC}

This section presents several constructions of optimal $(n,k,r,t)_i$ LRCs. These constructions mainly depend on the formation of the parity check matrix $H$ with the help of two different Cayley tables of a finite field.

\subsection{Construction 1}\label{Construction 1} Let us consider the Cayley table of $({\F_{p^m}}, +)$. After removing the operative row and the operative column from the table, we can assume it to be a square matrix $T^+ = (a_{ij})$ of order $p^m$. Using this matrix $T^+$, we now build a square matrix $M_{c} = (b_{ij})$ of order $p^m$ for each $c \in {\F_{p^m}}$ in the following way 
\[
M_c = \left\{
\begin{array}{ll}
b_{ij} = 1, &\text{if}~~a_{ij} = c~~\text{in}~~T^+,\\
b_{ij} = 0, &\text{otherwise}.
\end{array}
\right.
\]
Likewise by considering the Cayley table of $({\F_{p^m}} , \star)$, we can again imagine it as a square matrix $T^{\star}$ of order  $p^m$, as we include $0$ in the Cayley table. Now, directly place the matrix $M_{c}$ where there is a $c$ in $T^{\star} $. It forms a square matrix $A$ of order $p^{2m}$.

\begin{prop}\label{Proposition 3}
The above constructed matrix $A$ satisfies the following two properties.
\begin{enumerate}
    \item Each row weight and column weight of $A$ is exactly $p^m$.
    \item $\lvert Supp(\Omega_i) \cap Supp(\Omega_j)\rvert \leq 1$, where $\Omega_i$ and $\Omega_j$ are any two distinct rows of $A$.
\end{enumerate}
\end{prop}

\begin{proof}
In the Cayley table of $({\F_{p^m}}, +)$, each element of $\F_{p^m}$ occurs exactly once in each row and each column. Thus, row and column weight of $M_c$ are exactly one for all $c\in \F_{p^m}$. Using this fact, the first property can be easily proved. Again from Lemma \ref{Lemma 1}, the support sets of any two distinct rows cannot intersect in two or more points. Hence the result follows.
\end{proof}
\begin{remark}
If $H=[A|I_{p^{2m}}]$, then $H$ forms a parity check matrix of a binary optimal $(2p^{2m}, p^{2m}, p^m, p^m)_i$ LRC with distance $p^m+1$ and rate $\frac{1}{2}$.
\end{remark}

\begin{example}\label{Example 1}
Consider the field  $\F_5$ with $5$ elements and its two Cayley tables, which are the following,

\begin{table}[H]
     \begin{minipage}{.3\textwidth}
     $   \begin{array}{l|*{5}{l}}
       + & 0 & 1 & 2 & 3 & 4 \\ \hline
        0 & 0 & 1 & 2 & 3 & 4 \\
        1 & 1 & 2 & 3 & 4 & 0 \\ 
        2 & 2 & 3 & 4 & 0 & 1 \\ 
        3 & 3 & 4 & 0 & 1 & 2 \\ 
        4 & 4 & 0 & 1 & 2 & 3 \\
         \end{array} $\vspace{3mm}
        \caption{}
        \label{Table 1}
        \end{minipage}
        \begin{minipage}{.3\textwidth}
        $   \begin{array}{l|*{5}{l}}
      \star & 0 & 1 & 2 & 3 & 4 \\ \hline
              0 & 0 & 0 & 0 & 0 & 0 \\
              1 & 0 & 1 & 2 & 3 & 4 \\ 
              2 & 0 & 2 & 4 & 1 & 3 \\ 
              3 & 0 & 3 & 1 & 4 & 2 \\ 
              4 & 0 & 4 & 3 & 2 & 1 \\
         \end{array} $\vspace{3mm}
        \caption{}
        \label{Table 2}
        \end{minipage}
\end{table}
From Table \ref{Table 1} and Table \ref{Table 2}, we can obtain the corresponding $T^+$ and $T^{\star}$ matrices which are given by
\[
T^+=
\begin{bmatrix}
       0 & 1 & 2 & 3 & 4 \\
       1 & 2 & 3 & 4 & 0 \\ 
       2 & 3 & 4 & 0 & 1 \\ 
       3 & 4 & 0 & 1 & 2 \\ 
       4 & 0 & 1 & 2 & 3 \\
\end{bmatrix},\quad
T^{\star}=
\begin{bmatrix}
       0 & 0 & 0 & 0 & 0 \\
       0 & 1 & 2 & 3 & 4 \\ 
       0 & 2 & 4 & 1 & 3 \\ 
       0 & 3 & 1 & 4 & 2 \\ 
       0 & 4 & 3 & 2 & 1 \\
\end{bmatrix}.
\]
Using the matrix $T^+$, we now build the following matrices for each element of $\F_5$.
\[
M_0=
\begin{bmatrix}
1 & 0 & 0 & 0 & 0 \\
0 & 0 & 0 & 0 & 1 \\
0 & 0 & 0 & 1 & 0 \\
0 & 0 & 1 & 0 & 0 \\
0 & 1 & 0 & 0 & 0 \\
\end{bmatrix},\quad
M_1=
\begin{bmatrix}
0 & 1 & 0 & 0 & 0 \\
1 & 0 & 0 & 0 & 0 \\
0 & 0 & 0 & 0 & 1 \\
0 & 0 & 0 & 1 & 0 \\
0 & 0 & 1 & 0 & 0 \\
\end{bmatrix},\quad
M_2=
\begin{bmatrix}
0 & 0 & 1 & 0 & 0 \\
0 & 1 & 0 & 0 & 0 \\
1 & 0 & 0 & 0 & 0 \\
0 & 0 & 0 & 0 & 1 \\
0 & 0 & 0 & 1 & 0 \\
\end{bmatrix},
\]

\[
M_3=
\begin{bmatrix}
0 & 0 & 0 & 1 & 0 \\
0 & 0 & 1 & 0 & 0 \\
0 & 1 & 0 & 0 & 0 \\
1 & 0 & 0 & 0 & 0 \\
0 & 0 & 0 & 0 & 1 \\
\end{bmatrix},\quad
M_4=
\begin{bmatrix}
0 & 0 & 0 & 0 & 1 \\
0 & 0 & 0 & 1 & 0 \\
0 & 0 & 1 & 0 & 0 \\
0 & 1 & 0 & 0 & 0 \\
1 & 0 & 0 & 0 & 0 \\
\end{bmatrix}.
\]
We then place $M_0, M_1, \dots, M_4$ accordingly using $T^{\star}$, which gives the following matrix $A$ of order $25\times 25$
\[
A=
\begin{bmatrix}
M_0 & M_0 & M_0 & M_0 & M_0 \\
M_0 & M_1 & M_2 & M_3 & M_4 \\
M_0 & M_2 & M_4 & M_1 & M_3 \\
M_0 & M_3 & M_1 & M_4 & M_2 \\
M_0 & M_4 & M_3 & M_2 & M_1 \\
\end{bmatrix}.
\]
Now, if we arrange $H=[A|I_{25}]$, then $H$ forms a parity check matrix of a binary optimal $(50, 25, 5, 5)_i$ LRC.
\end{example}
\subsection{Construction 2}\label{Construction 2} First, we take the matrix $A$ from Construction \ref{Construction 1}. By using the matrix $A$, we build a square matrix $A'$ of order $(p^{2m}+p^m+1)$ in the following manner,
\[A'=\left[
\begin{array}{cc}
\multicolumn{1}{c|}{\textbf{1}_{p^m + 1}}    & \multicolumn{1}{c}{\textbf{0}_{p^{2m}}} \\
\hline

 & \multicolumn{1}{|c}{ }\\
 I_{p^m+1}\otimes \textbf{1}_{p^m}^T
 & \multicolumn{1}{|c}{I_{p^m}\otimes\textbf{1}_{p^m}}\\
\cline{2-2}
 &\multicolumn{1}{|c}{A}
\end{array}\right],\]

\begin{prop}\label{Proposition 4}
The above constructed matrix $A'$ satisfies the following two properties.
\begin{enumerate}
    \item The weight of each row and column of $A'$ is exactly $p^m+1$.
    \item $\lvert Supp(\Omega_i) \cap Supp(\Omega_j)\rvert \leq 1$, where $\Omega_i$ and $\Omega_j$ are any two distinct rows of $A'$.
\end{enumerate}
\end{prop}

\begin{proof}
The proof follows from Proposition \ref{Proposition 3} and the properties of the Kronecker product.
\end{proof}

\begin{remark}
If $H = [A'|I_{p^{2m}+p^m+1}]$, then $H$ becomes a parity check matrix of a binary optimal $\big(2(p^{2m}+p^m+1), p^{2m}+p^m+1, p^m+1, p^m+1\big)_i$ LRC with distance $p^m+2$ and rate $\frac{1}{2}$.
\end{remark}

\begin{example}
First, consider the square matrix $A$ of order 25 of Example \ref{Example 1}. Accordingly, the required matrix $A'$ of order $31\times31$ is
\[A'=\left[
\begin{array}{cc}
\multicolumn{1}{c|}{\textbf{1}_{6}}     & \multicolumn{1}{c}{\textbf{0}_{25}} \\
\hline
 & \multicolumn{1}{|c}{ }\\
 I_{6}\otimes\textbf{1}_{5}^T & \multicolumn{1}{|c}{I_{5}\otimes\textbf{1}_{5}}\\
\cline{2-2}
 &\multicolumn{1}{|c}{A}
\end{array}\right].\]
Now, if we set $H=[A'|I_{31}]$, then $H$ forms a parity check matrix of a binary optimal $(62, 31, 6, 6)_i$ LRC.
\end{example}

The importance of the above two constructions and this parity restricted class of $(n, k, r, t)_i$ LRCs are illustrated in the succeeding remarks.

\begin{remark}
The above two optimal $(n, k, r, t)_i$ LRC can be constructed for any finite field $\F_{p^m}$, in that case, we just need to replace the $1$'s of $A$ and $A'$ by any nonzero element of $\F_{p^m}$. The said technique is feasible only for this class of $(n, k, r, t)_i$ LRCs as it contains only one parity symbol in the repair set.
\end{remark}

\begin{remark}
Beside these discussions, the said technique of constructing such matrices $A$ and $A'$ can be used to form the explicit incidence matrices of finite projective planes of order $p^m$ and $p^m+1$ respectively in an easier way. 
\end{remark}

\subsection{Construction 3} For an $(n , k, r, t)_i$ LRC, as $d \geq t+1$, the above two optimal parameters have trivial minimum distance. We now provide an approach for constructing binary LRCs having non trivial optimal distance. Some of these can be found in \cite{Hao4} for some particular values. However, we present our construction vividly in a more general setting using those two matrices $A$ and $A'$. In this context, we start with the following observation.

\begin{prop}\label{Proposition 5}
Let us consider the matrix $H$ in the following manner
\[H= \left[\dfrac{A}{\textbf{1}_{p^{2m}}}\vline I_{p^{2m} +1}\right],\]
where $A$ is the chosen matrix from Construction \ref{Construction 1}, then $H$ forms a parity check matrix of an optimal binary $(2p^{2m} + 1, p^{2m}, p^m, p^m)_i$ LRC with distance $p^m + 2$ for an even $p$.
\end{prop}

\begin{proof}
We show that any $p^m +1$ columns of $H$ are linearly independent. First of all, it can be easily proved that any $p^m$ columns of $A$ are linearly independent. Now, as $p$ is even, thus any $p^m +1$ columns of $\left[\dfrac{A}{\textbf{1}_{p^{2m}}}\right]$ are linearly independent. Again support sets of any two distinct rows of $A$ intersect in at most one position, so it implies that the support sets of any two columns of $A$ also intersect in at most one position. Moreover, the weight of any column of $A$ is exactly $p^m$. Let $w$ be the weight of any nonzero linear combination of any $s~(<p^m +1)$ columns of $\left[\dfrac{A}{\textbf{1}_{p^{2m}}}\right]$, then 
\[
w \geq \left\{
\begin{array}{ll}
s\cdot p^m-2\cdot{s \choose 2} +1 = s(p^m -s) +s +1, & \text{if~~} s ~~\text{is~~odd};\\\\
s\cdot p^m-2\cdot{s \choose 2}  = s(p^m -s) +s, & \text{if~~} s ~~\text{is~~even}.
\end{array}
\right.
\]
Now, if we choose $p^m +1$ columns of $H$, $s$ from $\left[\dfrac{A}{\textbf{1}_{p^{2m}}}\right]$ and the remaining $p^m +1 -s$ columns from $I_{p^{2m} +1}$, then these columns are linearly independent as $w > p^m +1 -s$, for any $s$. Again, there always exists $p^m + 2$ linearly dependent columns due to the identity matrix and the column weight of $\left[\dfrac{A}{\textbf{1}_{p^{2m}}}\right]$. Hence, the minimum distance of the constructed code is $p^m + 2$.  
\end{proof}

The next result is obtained for an odd prime $p$. We proceed in a similar approach like above, and here the $A'$ matrix of Construction \ref{Construction 2} is being used. 

\begin{prop}\label{Proposition 6}
Let us consider the following matrix 
\[
H= \left[\dfrac{A'}{\textbf{1}_{p^{2m}+p^m+1}}\vline I_{p^{2m}+p^m+2}\right],\]
where $A'$ is the chosen matrix from Construction \ref{Construction 2}, then $H$ forms a parity check matrix of an optimal binary $\big(2(p^{2m} + p^m +1) +1, p^{2m} +p^m +1, p^m+1, p^m+1\big)_i$ LRC with distance $p^m + 3$ for an odd $p$.
\end{prop}

\begin{proof}
The proof is quite similar to Proposition \ref{Proposition 5}.
\end{proof}

\begin{remark}
The minimum distance in the above two constructions can be extended further by considering the similar pattern and using the rank of the matrices $A$ and $A'$. However, in that scenario, some other suitable rows are to be added in $H$.
\end{remark}
\section{property of $q$-ary optimal $(n, k, r, t)_i$ LRC}
This section proposes a generalized property of a $q$-ary optimal $(n, k, r, t)_i$ LRC. It can be observed that \cite[Theorem 1]{Hao} is a particular case of the following result with $t=1$.
\begin{thm}
Suppose $\C$ is an $(n, k, r, t)_i$ LRC having distance $d=n-k-\left\lceil\frac{kt}{r}\right\rceil+t+1$ and $H\equiv\left[\dfrac{P_1}{P_2}\vline I_{n-k} \right]$ is the standard parity check matrix of $\C$. Let $H_1$ and $H_2$  be two $m_1\times n_1$ and  $m_2\times n_2$ submatrices acquired from $H$ by deleting any fixed $\left\lceil\frac{kt}{r}\right\rceil-t$ and $\left\lceil\frac{kt}{r}\right\rceil-t-1$ rows among the first $l$ rows and all the columns whose coordinates are covered by the supports of those $\left\lceil\frac{kt}{r}\right\rceil-t$ and  $\left\lceil\frac{kt}{r}\right\rceil-t-1$ rows respectively. Then $H_1$ and $H_2$ have full row rank with $m_1=d-1$ and $m_2=d$. Moreover, the linear codes $\C_1 \equiv[n_1, k_1, d_1]$ and $\C_2 \equiv[n_2, k_2, d_2]$ with, respectively, $H_1$ and $H_2$ as their parity check matrices, accordingly, become $q$-ary MDS code and almost MDS code.
\end{thm}
\begin{proof}
First of all, it is easy to verify that 
\[m_1 = n-k-\left\lceil\frac{kt}{r}\right\rceil+t = d-1,\]and \[m_2 = n-k-\left\lceil\frac{kt}{r}\right\rceil+t+1=d.\]

Let $\gamma_1$ and  $\gamma_2$ be the number of columns covered by the omitted $\left\lceil\frac{kt}{r}\right\rceil-t$ and $\left\lceil\frac{kt}{r}\right\rceil-t-1$ rows successively. Then $\gamma_1\leq k+\left\lceil\frac{kt}{r}\right\rceil-t$ and $\gamma_2\leq k+\left\lceil\frac{kt}{r}\right\rceil-t-1$, as $H$ is in standard form. Thus, we have \[n_1\geq n-(k+\left\lceil\frac{kt}{r}\right\rceil-t)=n-k-\left\lceil\frac{kt}{r}\right\rceil+t= d-1 = m_1,\]

\[n_2\geq n-(k+\left\lceil\frac{kt}{r}\right\rceil-t-1)=n-k-\left\lceil\frac{kt}{r}\right\rceil+t+1= d = m_2.\] Now, it follows that,  $H_1\equiv\left[\dfrac{X_1}{Y_1}\vline I_{m_1} \right]$ and  $H_2\equiv\left[\dfrac{X_2}{Y_2}\vline I_{m_2} \right]$. Hence, we can say that $H_1$ and $H_2$ have full row rank with $m_1=d-1$ and $m_2=d$.\\
For the later part, at first we notice that $\left\lceil\frac{kt}{r}\right\rceil-t-1$ is non-negative, as $\left\lceil\frac{kt}{r}\right\rceil>t$ and if $\left\lceil\frac{kt}{r}\right\rceil=t+1$, then $H_2=H$. Now, to prove the main part, here we use the fact that $d\leq d_2\leq d_1$, since the above entries of those $n_2$ or $n_1$ columns of $H$, which resemble to the columns of $H_2$ or $H_1$, are all zeros. Again from the Singleton bound for the code $\C_1$, $d_1\leq m_1+1=d$.\\Combining both, we get $d_2 = d_1 = d$. Therefore, it is easy to observe that $\C_1$ forms a $q$-ary MDS code and $\C_2$ becomes a $q$-ary almost MDS code, as the Singleton defect of $\C_2$ is $1$.
\end{proof}
\section{Conclusion}
We study $(n, k, r, t)_i$ LRC through the standard parity check matrix approach. We provide the structural form of the parity check matrix for an $(n, k, r, t)_i$ LRC, which helps to discover different intrinsic algebraic properties of the code. We disclose how the membership matrix, local groups (subsets), and the parity check matrix are interlinked. Besides that, we prove some bounds by utilizing the parity check matrix tools. Additionally, precise constructions of some optimal codes are placed using the parity check matrix, which also performs well for non-binary cases. Finally, we give a generalized result on optimal $(n, k, r, t)_i$ LRC.

\section*{Acknowledgement}
We are thankful to Dr. Saikat Roy, Research Associate, Department of Mathematics, IIT Bombay, India for his constructive suggestions which improved the overall outfit of the paper. The first author of the paper would like to thank the Ministry of Education (MoE), Government of India, for financial support to carry out this research. This work is also supported by Science and Engineering Research Board (SERB), DST, India, Grant No. MTR/2021/000611.

\end{document}